\documentclass[reqno,11pt]{amsart}
\usepackage{amsmath,amssymb} 
\usepackage{amsthm, amsfonts, amsgen, stmaryrd,enumerate}
\usepackage[english]{babel}
\usepackage{fullpage}
\usepackage[centertags]{amsmath}
\usepackage{indentfirst}
\usepackage{fancyhdr}
\usepackage[dvips]{graphicx}
\usepackage{psfrag}
\numberwithin{equation}{section}
\usepackage[latin1]{inputenc}

\usepackage{color}

\allowdisplaybreaks[4]

\title[Warped product and Yang-Mills equations]{Warped Products and Yang-Mills equations on non commutative spaces}
\date{21 march 2014}
\author{Alessandro Zampini}
\address{Universit\'e du Luxembourg, Facult\'e des Sciences, de la Technologie et de la Communication, Mathematics Research Unit -- 6, rue Richard Coudenhove-Kalergi, L-1359 Luxembourg}
 \email{alessandro.zampini@uni.lu}

\newtheorem{prop}{Proposition}[section]
\newtheorem{lem}[prop]{Lemma}

\theoremstyle{definition}

\newcommand{\R}{\mathbb{R}}
\newcommand{\C}{\mathbb{C}}
\newcommand{\N}{\mathbb{N}}

\newcommand{\bS}{\mathbb{S}}
\newcommand{\CP}{\mathbb{C}\mathrm{P}}
\newcommand{\Q}{\mathbb{H}}
\newcommand{\A}{\mathcal{A}}
\newcommand{\E}{\mathcal{E}}
\newcommand{\mc}{\mathcal}

\newcommand{\openone}{\leavevmode\hbox{\small1\kern-3.8pt\normalsize 1}}
\newcommand{\dd}{\mathrm{d}}
\newcommand{\id}{\mathsf{id}}
\renewcommand{\S}{Sec.~}

\newcommand{\beq}{\begin{equation}}
\newcommand{\eeq}{\end{equation}}
\newcommand{\nn}{\nonumber}

\newcommand{\ca}{\mathcal{A}}

\newcommand{\cu}{\mathcal{U}}        
\newcommand{\SU}{\mathrm{SU}_q(2)}  
\newcommand{\ASU}{\ca(\mathrm{SU}_q(2))}  
\newcommand{\su}{\cu_q(\mathfrak{su}(2))}  

\newcommand{\eps}{\varepsilon}      

\newcommand{\hs}[2]{\left\langle #1,#2\right\rangle}  

\newcommand{\lt}{{\triangleright}}    
\newcommand{\rt}{{\triangleleft}}
\newcommand{\IC}{{\mathbb C}} 
\newcommand{\IR}{{\mathbb R}} 

\DeclareMathOperator{\U}{U}       

\begin{document}

\thispagestyle{empty}

\begin{abstract}
This paper presents a non self-dual solution of the Yang-Mills equations on a non commutative version of the classical $\R^4\smallsetminus\{0\}$, so generalizing the classical meron solution first introduced by de Alfaro, Fubini and Furlan in 1976. The basic tool for that is a generalization to non commutative spaces of the classical notion of warped products between metric spaces. 
\end{abstract}

\maketitle
\tableofcontents

\section*{Introduction}

\noindent 

The subject of Yang-Mills theory for projective modules over a $C^*$-algebra was introduced in \cite{CR87}, where the example of the non commutative torus was analyzed.  This originated an intense activity, which was mainly  focused on obtaining self-dual (or anti self-dual) solutions -- \emph{instantons} or \emph{anti-instantons} -- of Yang-Mills equations on suitable non commutative spaces. 

A class of them  is obtained by the  cocycle quantization associated with a torus action on the corresponding classical ones. This procedure is referred to as  $\theta$-deformation, the main example being the quantization of the $4$-dimensional sphere (see e.g.~\cite{BLvS12} and references therein).
The ADHM construction, which gave a complete description of instantons on $\bS^4$,
 has been adapted to several noncommutative spaces, among them the toric deformations $\bS^4_\theta$ of the $4$-sphere \cite{BL09}, the Moyal and toric deformations of $\R^4$ \cite{BvS10},  $\CP^2_\theta$ \cite{CLS11}.
General properties (e.g.~the infinitesimal structure) of the moduli space of
instantons on projective modules over an arbitrary toric deformation $M_\theta$
have been studied in \cite{BLvS12} (but notice that no explicit formula for
instantons is known with such a generality, even classically).
Reducible and irreducible instantons on $\CP^2_q$ have been indeed studied respectively
in \cite{DL09} and \cite{DL13}.
Beside this effort in obtaining  self-dual solutions, which provide -- in the classical setting -- informations about the geometry of the space itself, 
the study of non self-dual solutions on quantum spaces is a rather unexplored topic.

The aim of the present paper is to give, in analogy with the classical case, a non-selfdual solution for the Yang-Mills equation on a quantum analogue of the classical $\R^4\smallsetminus\{0\}$.

In order to shape the classical setting of the paper, consider the class of oriented spheres ${\bS}^n$ for $n\geq 2$ with their standard Riemannian metrics and the corresponding unique spinor bundles over them, namely the principal fibrations ${\rm Spin}(n+1)/{\rm Spin}(n)\to {\bS}^n$. The Levi-Civita connection associated to the standard Riemannian metric on ${\bS}^n$ can be lifted to a spinor connection, that is a connection on the principal bundle: it gives \cite{Lan86} a solution of the Yang-Mills equations on ${\bS}^{n}$.  If such a  spinor connection is transformed, via a stereographic projection,  into a connection on the Euclidean space $\IR^n$, then the resulting connection is a solution of the Yang-Mills equations on $\IR^n$ if and only if $n=4$.

For $n=2$ one has ${\rm Spin}(3)\,=\,{\rm SU(2)}\,=\,{\bS}^3$ and ${\rm Spin}(2)\,=\,\U(1)$; this spinor bundle gives the Hopf bundle ${\bS}^3/\U(1)\,=\,{\bS}^2$ and the connection above describes the monopole of lowest strenght. For $n=3$ one has 
${\rm Spin}(4)\,=\,{\rm SU(2)}\times{\rm SU}(2)\,=\,{\bS}^3\times{\rm SU}(2)$ and ${\rm Spin}(3)\,=\,{\rm SU}(2)$; the spinor connection $A$ on the bundle 
$({\bS}^3\times{\rm SU}(2))/{\rm SU}(2)\,=\,{\bS}^3$ turns out to be $A\,=\,(1/2)\omega$, where $\omega$ is the Maurer-Cartan form on ${\rm SU(2)}$. If $\pi\,:\,\R^4\backslash\{0\}\,\to\,\bS^3$ is the natural euclidean projection, then $\pi^*A$ is the \emph{meron} solution of Yang-Mills equations in four dimensions \cite{DFF76}.

The problem of introducing a meaningful spin connection on a wide class of quantum groups has been approached in \cite{HS97, Hec03} in the formalism of bicovariant bimodules, and in \cite{Maj99, Maj02, Maj03} using the braided formalism for the exterior algebras associated to universal and non universal --- still bicovariant --- differential calculi.

This paper presents  a generalization of the meron solution to the quantum group setting. If the quantum group $\SU$ is equipped with the widely studied three dimensional left covariant Woronowicz' calculus \cite{Wor87}, then, using the formalism developed in \cite{Hec01} a Maurer-Cartan $\omega$ form can be defined although the differential calculus is not bicovariant \cite{Wor89} and a suitable Cartan-Killing metrics can indeed be introduced with a meaningful Hodge operator. The vector potential $\mathfrak{A}\,=\,(1/2)\omega$ gives a solution for the corresponding Yang-Mills equations on $\SU$. 

Although our interest in the present paper concerns the structure of Yang-Mills equations, it is interesting to notice that the properties of the connection $\mathfrak{A}$ suggest to see it as a possible spin connection on $\SU$ with a 
three dimensional differential calculus (in this respect we recall \cite{BM11} for another  example of a spin connection --- coming from the introduction of suitable braiding operators --- on the same exterior algebra). We nonetheless leave this topic to a further development and extend to the product of non commutative spaces the  
the notion of almost product metrics, and explicitly study the example of the Hopf $*$-algebra injection $i:\ASU\,\hookrightarrow\,\A(GL_q(1, \Q))$ (the quantum group of invertible quaternions) showing that the connection $\mathfrak{A}$ can be mapped into a solution of the Yang-Mills equation on $\R^4_q\backslash\{0\}$. Such a solution is what we denote by \emph{quantum meron}, and has no instantonic (or anti-instantonic) character with respect to the natural Hodge duality operator.

The plan of the paper is the following.
In section \ref{sec:1}, we give a review of Yang-Mills theory for projective modules over a $*$-algebra.
In section \ref{sec:YMS3}, we recall the construction of \cite{Lan86},
and explain how the notion of  warped product between manifolds allows to obtain the meron solution on $\R^4\smallsetminus\{0\}$. 
In section \ref{sec:3}, we extend the notion of warped product to noncommutative spaces and link it to the study of Yang-Mills theory.
In section \ref{s:YMq}, we work out  a solution of the Yang-Mills equations on   $\bS^3_q$ and apply the warped product formalism to obtain a solution of them on $\R^4_q\smallsetminus\{0\}$. 

\subsection*{Acknowledgments}
Studying this problem was a nice suggestion of Giovanni Landi, to whom I feel gratitude also for the constant and precious feedback. I should like to thank Francesco D'Andrea, who helped me in understanding several passages of this analysis, as well as Aiyalam P. Balachandran and Giuseppe Marmo for the  various interesting discussions we had. It is indeed a pleasure for me to thank Norbert Poncin for his mentorship and the University of Luxembourg for the support.

\section{Yang-Mills equations on projective modules}
\label{sec:1}

In the classical framework of differential geometry, consider a  compact $N$-dimensional  smooth manifold equipped with a metric tensor $(M,\,g)$ and a smooth $\Bbbk$-vector bundle $E\,\to\,M$ on it (here $\Bbbk\,=\,\R$ or $\C$). The set $\Gamma(E)$ of smooth sections of the vector bundle is a $A\,=\,C^{\infty}(M, \Bbbk)$-module, where $A$ is the algebra of smooth $\Bbbk$-valued functions on $M$.  The Serre-Swan theorem proves that this assertion can be promoted to a  categorical equivalence.  A right (or left)  module $\Gamma$ over $C^{\infty}(M, \Bbbk)$ is finitely generated projective if and only if  a smooth $\Bbbk$-vector bundle $E\,\to\,M$ exists, such that $\Gamma\,\simeq\,\Gamma(E)$, and that two vector bundles are isomorphic if and only if the corresponding modules of sections are. 

Denoting by $(\Lambda^{\bullet}, \dd)$ the differential graded exterior algebra of forms on $M$, a linear connection is a graded $A$-module derivation $\nabla\,:\,\Gamma(E)\,\otimes_A\Lambda^{\bullet}\,\to\,\Gamma(E)\,\otimes_A\Lambda^{\bullet+1}$, that  describes what is usually referred to as  the vector potential of a gauge field, while elements in $\Gamma(E)$ are the matter fields.  The $A$-module map $F\,=\,\nabla^2\,:\,\Gamma(E)\,\otimes_A\Lambda^{\bullet}\,\to\,\Gamma(E)\otimes_A\Lambda^{\bullet+2}$ is the corresponding curvature: if $D$ suitably extends in terms of graded commutators to $A$-module  maps  the action of the covariant derivative $\nabla$, we can write
\begin{align}
D\,F&=\,0, \\
D(\ast_H\,F)&=\,0, \label{clYM}
\end{align}
where the first line represents the Bianchi identity and the second line the homogeneous Yang-Mills equations. Notice that  $\ast_H\,:\,\Lambda^k\,\to\,\Lambda^{N-k}$ represents the natural extension as an $A$-module map of the  Hodge duality operator\footnote{It satisfies the identity  $\ast_H^2\,=\,({\rm sgn}\,g)\,(-1)^{k(N-k)}$ with $({\rm sgn}\,g) $ the signature of the metrics.} corresponding to the metric tensor $g$ on $M$.

The Serre-Swan theorem suggests that the natural framework for an algebraic formulation of a gauge theory is that of projective (finitely generated) modules over a ---  possibly non commutative --- algebra $A$,  equipped with a finite dimensional differential calculus. In the following pages we shall review how it is possible to introduce on such modules a covariant derivative, i.e. a linear connection, and describe how a hermitian structure over a (say) right $A$-module allows to introduce a Hodge duality and then Yang-Mills equations.

\subsection{Differential graded algebras and connections}
\label{dgac}

Let $\A$ be an associative $*$-algebra over $\Bbbk$. A differential calculus on it is a differential graded algebra,
(DGA -- for this part we refer to \cite{DV01}), that is  the datum $(\Lambda^\bullet,\dd)$
of a graded associative algebra $\Lambda^\bullet=\bigoplus_{k\geq 0}\Lambda^k$ and a
map $\dd:\Lambda^\bullet\to\Lambda^{\bullet+1}$, which is a coboundary, $\dd^2=0$, and 
a graded derivation:
\begin{equation*}
\dd(\alpha\beta)=(\dd\alpha)\beta+(-1)^{\mathrm{dg}(\alpha)}\alpha(\dd\beta) \;.
\end{equation*}
A differential calculus over an associative algebra $\A$ is a DGA $(\Lambda^\bullet,\dd)$
such that $\Lambda^0=\A$ and $\Lambda^{k+1}=\mathrm{Span}\{a\dd\omega,\,a\in\A,\,\omega\in\Lambda^k\}$
for all $k\geq 0$. A differential calculus is \emph{finite-dimensional} with dimension $n\in\N$
if $\,\Lambda^k=0$ for all $k>n$. We further assume that $\Lambda^k$ is a free $\A$-bimodule of finite rank.

A graded involution on $\Lambda^\bullet$ is a degree zero map $*:
\Lambda^\bullet\to\Lambda^\bullet$ satisfying $(\omega^*)^*=\omega$ and 
$(\alpha\beta)^*=(-1)^{\mathrm{dg}(\alpha)\mathrm{dg}(\beta)}\beta^*\alpha^*$. 
If $\dd$ is a coboundary resp.~a graded derivation (i.e.~it satisfies the graded Leibniz rule),
then the map $\omega\mapsto \dd_*\omega:=(\dd\omega^*)^*$ is a coboundary resp.~a graded derivation.
A differential calculus is called \emph{real} or \emph{$*$-calculus} if there is a graded involution,
extending the involution of $\A=\Lambda^0$, such that $\dd_*=\dd$.

For a given choice of differential calculus $(\Lambda^\bullet,\dd)$ on a $*$-algebra $\A$,
we call a \emph{connection} on a right $\A$-module $\E$ a $\Bbbk$-linear map $\nabla:\E\otimes_{\A}\Lambda^\bullet\to
\E\otimes_{\A}\Lambda^{\bullet+1}$ satisfying the graded Leibniz rule:
\begin{equation}
\label{defc}
\nabla(\eta\omega)=(\nabla\eta)\omega+(-1)^{\mathrm{dg}(\eta)}\eta\,\dd\omega \;,
\end{equation}
for all $\eta\in\E\otimes_{\A}\Lambda^{{\rm dg}\,\eta}$ and $\omega\in\Lambda^\bullet$.
This generalizes the notion of linear connection of differential geometry.
Connections on left modules are defined in a similar way.
Let
\begin{equation}\label{eq:gamman}
\mathcal{R}^n=
\mathrm{Hom}_{\Lambda^\bullet}(\E\otimes_{\A}\Lambda^{\bullet}, \E\otimes_{\A}\Lambda^{\bullet+n})
\end{equation}
be the space 
of degree $n$ right $\Lambda^\bullet$-module endomorphisms of $\E\otimes_{\A}\Lambda^{\bullet}$.
Then $\mathcal{R}^\bullet:=\bigoplus_{n\geq 0}\mathcal{R}^n$ is a graded associative unital algebra, with product
given by the composition of endomorphisms, and $\mathcal{R}^0$ is a unital subalgebra isomorphic to
$\mathrm{End}_{\A}(\E)$. Notice that $\mathcal{R}^\bullet$ 
is a $\mathcal{R}^0$-bimodule, but in general it is not a right $\A$-module (unless the algebra is commutative).
From \eqref{defc} one sees \cite{Lod92} that a connection is completely determined by its restriction on $\E$ (and then extended by the Leibniz rule). 
It  follows that any two connections $\nabla$ and $\nabla'$ differ by an element
\mbox{$\nabla-\nabla'\in\mathcal{R}^1$}. 
If $\E$ is a projective $\A$-module --- that is $\E\,=\,p\,\A^k$ where $p\,:\,\A^k\,\to\,\A^k$ with $p^2\,=\,p\,=\,p^{\dagger}$ --- then one can see that connections exist and  that ${\rm Hom}_{\A}(\E, \E\otimes_{\A}\Lambda^1)\,\simeq\,\mathbb{M}_k(\A)\,\otimes_{\A}\Lambda^1$, so on them one can write  
\begin{equation}
\label{ccu}
\nabla\,=\,\nabla_0\,+\,\alpha
\end{equation}
where $\nabla_0\,=\,p\dd$ is the Grassmann (canonical) connection and   $\alpha\,\in\,\mathcal{R}^1$  can be represented as a matrix whose entries are elements in $\Lambda^1$. 
One also has that the map $\nabla^2\,=\,\nabla\,\circ\,\nabla\,:\,\E\,\otimes_{\A}\Lambda^{\bullet}\,\to\,\E\,\otimes_{\A}\Lambda^{\bullet+2}$ is $\Lambda^{\bullet}$-linear. Its restriction to $\E$, denoted by 
\begin{equation}
\label{dcm}
F\eta=\nabla^2\eta\qquad\forall\,\eta\in\E\;.
\end{equation}
is called the curvature of the connection $\nabla$.  When $\E$ is projective, the curvature $F$ is an element in $\mathbb{M}_{k}(\A)\otimes_{\A}\Lambda^2\,\simeq\, \mathcal{R}^2$. 
The connection is called
\emph{flat} if $F=0$ (e.g.~$\dd$ is a flat connection on each $\Lambda^k$).
As usual for graded algebras, we define the commutator of
$A\in\mathcal{R}^j$ and $B\in\mathcal{R}^k$ as
\begin{equation*}
[A,B]=A\circ B-(-1)^{jk}B\circ A
\end{equation*}
and extend it by bilinearity. If $A$ has odd degre, one easily checks
that $[A,\,.\,]$ is a graded derivation on $\mathcal{R}^\bullet$. Although
in general $\nabla$ is not an element of $\mathcal{R}^1$ (it is a degree $1$ endomorphism,
but it is not right $\Lambda^\bullet$-linear),
nevertheless the operator
\begin{equation*}
D=[\nabla,~]\,:\,\mathcal{R}^n\to\mathcal{R}^{n+1}
\end{equation*}
is well defined ($D(A)$ is
right $\Lambda^\bullet$-linear for any $A\in\mathcal{R}^\bullet$)
and it is a graded derivation on $\mathcal{R}^n$.
Within this formalism, the \emph{Bianchi identity}
\begin{equation}
\label{Bi}
D(F)=0
\end{equation}
is simply the statement that $[\nabla,F]=\nabla\circ\nabla^2-\nabla^2\circ\nabla=\nabla^3-\nabla^3=0$.

\subsection{Hermitian and self-dual modules}
\label{ss:hss}
We have seen how a connection -- and the corresponding curvature -- is defined on a module over an associative algebra, how it allows to introduce a covariant derivative so to have a consistent algebraic version of the Bianchi identity, which is the first relation in \eqref{clYM}. In order to introduce a Hodge duality operator we describe now self-dual modules.

A \emph{Hermitian structure} on a right $\A$-module $\E$ is a sesquilinear map
$(\,,\,)_{\E}:\E\times\E\to\A$ satisfying $\,(\eta a,\xi b)_{\E}=a^*(\eta,\xi)_{\E}b\,$ and
$\,(\eta,\xi)_{\E}^*=(\xi,\eta)_{\E}\,$ for all $\eta,\xi\in\E $ and $a,b\in\A$,
together with  $\,(\eta,\eta)_{\E}>0\,$ for all $\eta\neq 0$.
In case $\A$ is a pre $C^*$-algebra (here by this we simply mean that it is a dense $*$-subalgebra of some $C^*$-algebra), $\E$ will be called a right \emph{pre-Hilbert $\A$-module}. Every finitely generated projective module has a canonical pre-Hilbert module structure, that is unique modulo an equivalence.

If there is a Hermitian structure on $\E$, we call the connection $\nabla$ to be \emph{Hermitian} if
\begin{equation*}
(\nabla\eta,\xi)_{\E}+(-1)^{\mathrm{dg}(\eta)}(\eta,\nabla\xi)_{\E}=\dd(\eta,\xi)_{\E}
\end{equation*}
for all $\eta,\xi\in\E\otimes_{\A}\Lambda^\bullet$. Here the Hermitian structure is extended from $\E$ to $\E\otimes_{\A}\Lambda^\bullet$ in a natural way: if
$\eta,\eta'\in\E$, $\omega\in\Lambda^{j}$ and $\omega'\in\Lambda^{k}$ we define
\begin{equation*}
(\eta\otimes\omega,\eta'\otimes\omega')_{\E\otimes_{\A}\Lambda^\bullet}:=\omega^*\,(\eta,\eta')_{\E}\,\omega'\,\in\Lambda^{j+k} \;.
\end{equation*}
Compatibility with the Hermitian structure is a mild restriction on connections: while the set of connections on $\E$ is a complex affine space, the set of Hermitian connections is 
a real affine space. If the connection is given by $\nabla\,=\,\dd\,+\,\alpha$ (see equation \eqref{ccu}), then it is hermitian if and only if 
\begin{equation}
\label{herc}
\alpha^{\dagger}\,+\alpha\,=\,0
\end{equation}
in $\mathbb{M}_k(\A)\,\otimes_{\A}\Lambda^1$.

Let $\E$ a pre-Hilbert right $\A$-module.
The set $\E'=\mathrm{Hom}_{\A}(\E,\A)$ of right $\A$-module maps $\phi:\E\to\A$
is itself a right $\A$-module, with module structure $(\phi\cdot a)(\eta):=a^*\phi(\eta)$
for all $a\in\A$, called the \emph{dual module} of~$\E$.

Given two modules $\E$ and $\mathcal{F}$ with Hermitian structures
$(\,,\,)_{\E}$ and $(\,,\,)_{\mathcal{F}}$, a map
$T:\E\to\mathcal{F}$ is called \emph{adjointable} if $\exists\;
T^*:\mathcal{F}\to\E$ such that $(T^*\eta,\xi)_{\E}=(\eta,T\xi)_{\mathcal{F}}$
for all $\eta\in\mathcal{F}$ and $\xi\in\E$. An adjointable map is always
a right $\A$-module map, but the converse is not true.

\begin{lem}\label{lemma:sd}
For a pre-Hilbert right $\A$-module $\E$ the following are equivalent:
\begin{itemize}\itemsep=0pt
\item[i)] the right $\A$-module map $\E\to\E'$, $\eta\mapsto 
(\eta,\,\cdot\,)_{\E}$ is surjective;
\item[ii)] any right $\A$-module map $\E\to\A$ is adjointable.
\end{itemize}
Each one of these two equivalent conditions implies that any right
$\A$-module map $\E\to\mathcal{F}$ is adjointable for any pre-Hilbert right
$\A$-module $\mathcal{F}$.
\end{lem}

\begin{proof}
Suppose that every right $\A$-module map $T:\,E\,\to\,\A$, i.e. any element $T\,\in\,\E^{\prime}$ is adjointable: there exists  a right $\A$-module map $T^*\,:\,\A\,\to\,\E$ such that the relation 
$$
(f, T(\varphi))_{\A}\,=\,(T^*(f),\varphi)_{\E}
$$
holds for any $f\,\in\,\A$ and $\varphi\,\in\,\E$. The lhs is defined as  $(f, T(\varphi))_{\A}\,=\,f^*T(\varphi)$, so for $f\,=\,1$ one has $T(\varphi)\,=\,(T^*(1),\varphi)_{\E}$, which proves i), provided one defines $\eta\,=\,T^*(1)$ in $\E$. 

If i) holds, then for any $T\,\in\,\E^{\prime}$ there exists an element $\tau_{T}\,\in\,\E$ such that $T(\varphi)\,=\,(\tau_T, \varphi)_{\E}$, and one can write $f\,T(\varphi)\,=\,(f\,\tau_T,\varphi)_{\E}$ for any $f\,\in\,\A$ and $\varphi\,\in\,\E$. The map $T^*$ is defined by $f\,\mapsto\,f\,\tau_T$.
\end{proof}

\noindent
A right module is \emph{self-dual} if it satisfies any one
of the equivalent conditions of Lemma \ref{lemma:sd}.
An adjointable map $\E\to\mc{F}$ is called a \emph{homomorphism} between the two pre-Hilbert modules, and the set of all such maps will be denoted $\mathrm{Hom}^*_{\A}(\E,\mc{F})$.
If $\E$ is self-dual, any right $\A$-module map is adjointable and hence a homomorphism.
Since the map $\E\to\E'$ is always injective ($(\eta,\eta)_{\E}=0$
implies $\eta=0$), a self-dual module is one that is isomorphic to its dual.

A Hermitian structure exists on any finitely generated projective module: if $\E=p\A^k$, with
$p=p^*=p^2\in M_k(\A)$ a projection, all Hermitian structure are equivalent
to the one obtained by restricting to $\E$ the standard Hermitian structure
on $\A^k$ given by
\begin{equation}\label{eq:hs}
(\eta,\xi)=\sum\nolimits_{i=1}^k\eta_i^*\xi_i
\end{equation}
for $\eta=(\eta_1,\ldots,\eta_k)$ and $\xi=(\xi_1,\ldots,\xi_k)\in\A^k$.

\begin{lem}
\label{lemsf}
The module $\E=p\A^k$ with Hermitian structure \eqref{eq:hs} is self-dual.
\end{lem}
\begin{proof}
Any $T\in\E'$ can be extended to a right $\A$-module map $T:\A^k\to\A$ by
declaring $T$ equal to zero on the complement $(1-p)\A^k$.
Assume $\A^k$ is self-dual, then $\exists\;\tau_T\in\A^k$ such that $T(\xi)=(\tau_T,\xi)$
for all $\xi$. Since \eqref{eq:hs} satisfies
\begin{equation}\label{eq:canhs}
(\tau_T,A\xi)=(A^*\tau_T,\xi)\;,\qquad \text{for any}\;A\in \mathbb{M}_k(\A)\;,
\end{equation}
then the restriction to $\E$ gives $T(p\xi)=(p\tau_T,p\xi)$ for any $p\xi\in\E$. 
Since $p\tau_T\in\E$, the module $\E$ is self-dual.
It is then enough to prove the lemma when $\E=\A^k$.

Any right $\A$-module
map $T:\A^k\to\A$ is of the form $T(\xi)=\tau_1a_1+\ldots+\tau_ka_k$, $\xi=(a_1,\ldots,a_k)\,\in\,\A^k$,
where $\tau_i\in\A$ is the image of the $i$-th element in the canonical basis of $\C^k$.
Called $\tau_T=(\tau_1^*,\ldots,\tau_k^*)$, then $T(\xi)=(\tau_T,\xi)$,
which concludes the proof.
\end{proof}

\noindent
Notice that a generating family $\{v_i\}_{i=1}^k$ for $p\A^k$
is given by the columns of the projection $p$.
By linearity, Hermitian structures on $p\A^k$ are uniquely
determined by the positive element $g=(g_{ij})\in M_k(\A)$ given by
$g_{ij}=(v_i,v_j)$. Since by construction $v_j=\sum_{i=1}^kv_ip_{ij}$,
the matrix $g$ satisfies $pg=gp=g$. The canonical Hermitian structure
corresponds to the choice $g=p$.

\subsection{The Hodge star operator}\label{sec:1.4.2}
Given the associative $*$-algebra $\A$ equipped with a $n$-dimensional differential calculus  $(\Lambda^\bullet,\dd)$, we have assumed that $\Lambda^k$
are right $\A$-modules. In order to define a right module map $\ast_H:\Lambda^k\to\Lambda^{n-k}$ analogue to the
Hodge star operator, a necessary condition is that $\Lambda^n$ is a rank $1$ free right module, with basis given by an element $\tau$ that we will call the \emph{volume form}, and that, as free modules, ${\rm rank}\,\Lambda^N\,=\,{\rm rank}\,\Lambda^{N-k}$.
Since $\Lambda^k$ is a free (right, say) $\A$-module,  from the previous lemma \ref{lemsf} we know that, via the  equation \eqref{eq:hs}
referred to a chosen fixed basis, $\Lambda^k$ has an Hermitian structure $(~,~)_k$ by which it turns to be self-dual. 
Via the volume form we may write 
the Hermitian structure  on $\Lambda^0$ as $(a,b)_0=a^*b\;\forall\;a,b\in\A$, while on 
$\Lambda^n\,\simeq\,\A$ we may similarly write $(\tau a,\tau b)_n=a^*b\;\forall\;a,b\in\A$.
Notice that by adjunction, $\tau^*$ is a basis of $\Lambda^n$ as a left module.
We would like then to define $\ast_H$ on arbitrary (complex) forms by the formula
\begin{equation}\label{eq:2.2}
\tau\,(\,\ast_H\hspace{1pt}\alpha,\beta)_{n-k}=\alpha^*\wedge\beta \;,
\end{equation}
for each $\alpha\in\Lambda^k$ and $\beta\in\Lambda^{n-k}$ (here we consider complex
algebras and complex forms, and have denoted via the standard $\wedge$ the product between forms). The following proposition holds. 
\begin{prop}\label{prop:list}~
\begin{enumerate}[i)] 
\item  There exists,  and it is unique,  a linear map $\ast_H:\Lambda^k\to\Lambda^{n-k}$
satisfying \eqref{eq:2.2}.

\item It is $\,\ast_H(1)=\tau\,$ and $\,\ast_H(\tau^{*})=1\,$.

\item It is $\ast_H$ is a right $\A$-module map if and only if $\tau$ is central
(i.e.~$\Lambda^n\simeq\A$ as a bimodule).

\end{enumerate}
\end{prop}

\begin{proof}
One has:
\begin{enumerate}[i)]
\item Since $\Lambda^n$ is a rank $1$ free right module and $\Lambda^k$ is
selfdual, from Lemma~\ref{lemma:sd} the linear map $\Lambda^k\to\mathrm{Hom}_{\A}(\Lambda^k,\Lambda^n)$,
$\eta\mapsto\tau\,(\eta,\,\cdot)_k$, is surjective.
Fixed $\alpha\in\Lambda^k$, the map $\varphi_\alpha:\Lambda^{n-k}\to\Lambda^n$ given by 
$\varphi_\alpha(\beta):=\alpha^*\wedge\beta$ is a right $\A$-module map,
i.e.~$\varphi_\alpha\in\mathrm{Hom}_{\A}(\Lambda^{n-k},\Lambda^n)$.
Hence, there exists and is unique an $\eta\in\Lambda^{n-k}$ such that 
$\varphi_\alpha(\beta)=\tau\,(\eta,\beta)_{n-k}$. We define  $\ast_H\alpha=\eta$
and notice that it satisfies \eqref{eq:2.2} by construction.
\item From \eqref{eq:2.2} we get $(\ast_H(1),\tau)_n=1$ and $(\ast_H(\tau^*),1)_0=1$. This implies ii).
\item For all $a\,\in\,\A,\,\alpha\,\in\,\Lambda^k,\,\beta\,\in\,\Lambda^{n-k}$:
\begin{equation}\label{eq:propiii}
\tau\,\bigl(\ast_H(\alpha a)-(\ast_H\alpha)a,\beta\bigr)_{n-k}=
[a^*,\tau](\ast_H\alpha,\beta)_{n-k} \;.
\end{equation}
The left hand side is zero if and only if, by non-degeneracy
of the Hermitian structure, $\ast_H(\alpha a)=(\ast_H\alpha)a$ for all $a,\alpha$,
i.e.~$\ast_H$ is a right module map. Choosing $\alpha=\tau^*$ and
$\beta=1$, then $(\ast_H\alpha,\beta)_0=1$ and the right hand side of
\eqref{eq:propiii} vanishes if and only if $[a^*,\tau]=0$ for all
$a^*\in\A$. Therefore, $\ast_H$ is a right module map if and only if  $\tau$
is central.
\end{enumerate}
\end{proof}

\noindent Left multiplication by $\A$ gives a map $\pi:\A\to\mathrm{End}_\A(\Lambda^\bullet)$.
In fact, since we are assuming that $\Lambda^\bullet$ is selfdual, every right $\A$-linear
endomorphism is adjointable, and $\pi$ maps to the $*$-algebra $\mathrm{End}^*_\A(\Lambda^\bullet)$.

\begin{prop}
$\ast_H$ is a left $\A$-module map if and only if $\pi$ is a $*$-representation.
\end{prop}

\begin{proof}
For all $a\,\in\,\A,\,\alpha\,\in\,\Lambda^k,\,\beta\,\in\,\Lambda^{n-k}$ one has:
\begin{equation*}
\tau\,\bigl(\ast_H(a\alpha),\beta\bigr)_{n-k}=
\alpha^*a^*\wedge\beta=
\tau\,(\ast_H\alpha,a^*\beta)_{n-k} \;,
\end{equation*}
that is $\bigl(\ast_H(a\alpha),\beta\bigr)_{n-k}=(\ast_H\alpha,a^*\beta)_{n-k}$.
Thus $\ast_H(a\alpha)=a\cdot {\ast_H}(\alpha)$ for all $a,\alpha$ if and only if,
for all $a\in\A$, $\pi(a)$ is an adjointable endomorphism with adjoint $\pi(a^*)$.
That is, $\pi:\A\to\mathrm{End}^*_\A(\Lambda^\bullet)$ is a $*$-representation.
\end{proof}

\noindent How is the duality map defined via the relation \eqref{eq:2.2} related to the Hodge duality operator defined in the setting of differential geometry?  With $(M,g)$ a  $N$-dimensional smooth manifold equipped with a metric tensor $g$, let $g\,=\,\sum_{ab}g^{ab}\,X_a\otimes X_b$ be the expression of the metrics in a coordinate chart system $(x^{i})_{i\,=\,1,\ldots,N}$, where $(\Lambda^{\bullet}, \dd)$ is the standard  $N$-dimensional exterior  differential calculus with $\Lambda^{N-k}\,\simeq\,\Lambda^{k}\,\simeq\,\A^{k(k-1)/2}$ as free $\A$-bimodules and the vector fields $X^a$ give the dual basis  to the 1-form basis $\dd x^a$, i.e. $i_{X_a}\dd x^b\,=\,\delta_a^b$. Define 
\begin{equation}
\label{clhs}
(\dd x^{a_1}\wedge\ldots\wedge\dd x^{a_k}, \dd x^{b_1}\wedge\ldots\wedge\dd x^{b_{k}})_k\,=\,
\frac{1}{k!}\,\sum_{\pi\,\in\,S_k}\,\sum_{\pi^\prime\,\in\,S_k} g^{\pi(a_1)\,\pi^{\prime}(b_1)}\,\cdots\,g^{\pi(a_k)\,\pi^{\prime}(b_k)}
\end{equation}
where $S_k$ is the permutation group of $k$ elements.  Assume the volume form is $\tau\,=\,\dd x^1\wedge\ldots\wedge\dd x^N$. One can then easily see that the duality operator $\ast_{H}$ \eqref{eq:2.2} corresponding to such a hermitian structure is the well known Hodge operator on $\Lambda^k$. It is moreover possible to prove that, if $T\,=\,\sum_{ab}T^{ab}X_a\otimes X_b$ is a contravariant tensor on the same coordinate chart, and one defines in terms of the components of $T$ a  hermitian product which is analogue to \eqref{clhs}, then the duality operator $\ast_T\,:\,\Lambda^{k}\,\to\,\Lambda^{N-k}$ (defined via \eqref{eq:2.2}) satisfies the identity $\ast_T^2\,=({\rm sgn}\, T)\,(-1)^{k(N-k)}$ -- with ${\rm sgn}\, T$ the signature of $T^{ab}$ -- if and only if it is symmetric (i.e. $T^{ab}\,=\,T^{ba})$ and invertible (i.e. $\det\,T^{ab}\,\neq\,0$). Moreover, one can see that the tensor $T$ is real (i.e. $T^{ab*}\,=\,T^{ab}$) if and only if $\ast_H(\xi^*)\,=\,(\ast_H\xi)^*$.

\subsection{The Yang-Mills equations}

Recall the general analysis performed in \ref{dgac}. In order to be able to introduce the analogue of the equation \eqref{clYM}, one needs to define the action of the Hodge duality operator on the curvature of a given connection, and they are given (see \eqref{dcm}) by elements in $\mathbb{M}_k(\A)\otimes_{\A}\Lambda^2$ for a projective module $\E\,=\,p\,\A^k$.   We set (adopting the same symbol, with a slight abuse of notation)
\begin{equation}
\label{cur}
\ast_H\,:\,\mathbb{M}_k(\A)\otimes_{\A}\Lambda^2
\,\xrightarrow{\;\id\,\otimes_{\A}\text{ $\ast_H$}}
\mathbb{M}_k(\A)\otimes_{\A}\Lambda^{n-2},
\end{equation}
indeed naturally extending the action of the Hodge duality previously presented. Notice that this definition is meaningful even if $\ast_H$ is not a bimodule map.
Following such a definition, we define the (homogeneous) Yang-Mills equation as
\begin{equation}
\label{qYM}
D(\ast_H\,F)=0 \;.
\end{equation}
If $n=4$ and $F$ is an eigenvector of $\ast_H$, the above equation is  automatically satisfied, due to the Bianchi identity \eqref{Bi}.

\section{Solutions of the Yang-Mills equations on $\bS^3$}
\label{sec:YMS3}
As a manifold $M$ consider the 3-sphere $\bS^3$. Since it is a group manifold, $SU(2)\,\simeq\,\bS^3$, its exterior algebra is globally defined. If we set $\A\,=\,C^{\infty}(\bS^3)$ to be the algebra of smooth complex valued functions on $\bS^3$, $\Lambda^1$ is a free $\A$-bimodule of rank 3 and  has a basis given by $\{\omega_a\}_{a\,=\,1,\ldots,3}$, where the complex conjugation acts as $\omega_a\,=\,\omega_a^*$. Such 1 forms are  dual to the set of vector fields $\{X_{a}\}_{a\,=\,1,\ldots,3}$  which give a realization of the Lie algebra $\mathfrak{su(2)}$, 
\begin{equation}
\label{geg}
[X_a, X_b]\,=\,\varepsilon_{abc}\, X_c. 
\end{equation}
The action of the exterior derivative on $\Lambda^{\bullet}$ is completely characterized by the Leibniz rule and the Maurer-Cartan relation
\begin{equation}
\label{mccl} 
\dd\,\omega_a\,=\,-\,\frac{1}{2}\,\varepsilon_{abc}\,\omega_{b}\wedge\omega_c
\end{equation}
On this manifold we consider the \emph{round} metrics -- the one induced by the euclidean metrics on $\R^4$ --  that is $g\,=\,\sum_a\omega_a\otimes\omega_a$ which coincides with the Cartan-Killing metrics on the dual of the Lie algebra. If the  volume form is chosen to be  $\tau\,=\,\omega_1\wedge\omega_2\wedge\omega_3$, the Hodge duality operator corresponding to \eqref{clhs} turns out to be the $\A$-bimodule map given on a basis of $\Lambda^{\bullet}$ by:
\begin{align}
\ast_H(1)\,=\,\tau&\qquad\qquad\ast_H(\tau)\,=\,1\nonumber \\
\ast_H(\omega_a)\,=\,\frac{1}{2}\,\varepsilon_{abc}\omega_b\wedge\omega_c,& \qquad\qquad\ast_H(\omega_a\wedge\omega_b)\,=\,\varepsilon_{abc}\omega_c
\label{Hcl}
\end{align}
It is possible to see (see \cite{weo}) that any vector bundle $\pi\,:\,E\,\to\,\bS^3$ is equivalent to a trivial one, so its space of section is a free module over $\A$. We consider then the rank two free bimodule $\E\,=\,\A^{2}$. Given the canonical hermitian structure on $\E$,  from the equation \eqref{herc} we know that the set of hermitian connections
$\mathcal{R}^1\,\simeq\,{\rm End}_{\A}\otimes_{\A}\,\Lambda^1$ (notice that this equivalence is valid only if $\A$ is commutative) can be written as
\begin{equation}
\label{r1c}
\mathcal{R}^1\,\ni\,\alpha\,=\,\sum_j\alpha_j\,\otimes\,T_j
\end{equation}
where $\alpha_j\,=\,\alpha_j^*\,\in\,\Lambda^1$ and $T_j\,=\,-T_j^*$ are the matrices giving the elements of the Lie algebra $\mathfrak{su(2)}$ in the fundamental representation:
\begin{equation}
T_1\,=\,\frac{1}{2}\,\left(\begin{array}{cc} 0 & i \\ i & 0\end{array}\right), \qquad
T_2\,=\,\frac{1}{2}\,\left(\begin{array}{cc} 0 & -1 \\ 1 & 0\end{array}\right), \qquad
T_3\,=\,\frac{1}{2}\,\left(\begin{array}{cc} i & 0 \\ 0 & -i\end{array}\right). \qquad
\label{Xge}
\end{equation}
In particular, we select 
\begin{equation}
\alpha\,=\,\lambda\,\omega_a\,\otimes\,X_a
\label{calpha}
\end{equation}
 with $\lambda\,\in\,\R$.  From  \eqref{dcm}, \eqref{cur}, \eqref{Hcl}, \eqref{mccl} one has:
\begin{align}
&F\,=\,\frac{\lambda}{2}\,(\lambda\,-\,1)\,\epsilon_{abc}\,\omega_a\,\wedge\,\omega_b\,\otimes\,X_c\,=\,\lambda\,(1\,-\,\lambda)\,\dd\omega_a\,\otimes\,X_a \label{ccu}\\
&\ast_H(F)\,=\,(\lambda\,-\,1)\,\alpha  \label{eqi}\\ 
&D(\ast_H \,F)\,=\,\lambda\,(\lambda\,-\,1)\,(\lambda\,-\,\frac{1}{2}\,)\,\epsilon_{abc}\,\omega_a\,\wedge\,\omega_b\,\otimes\,X_c \label{yms}
\end{align}
This shows that, except the trivial cases $\lambda\,=\,0,1$, the choice $\lambda\,=\,1/2$ gives a non trivial solution to the Yang-Mills equations. The previous result depends only on the commutation relation of the Lie algebra $\mathfrak{su(2)}$ and the Hodge duality. 

One can check that for $\lambda=1/2$ the connection in \eqref{calpha} has zero torsion, and is then the Levi-Civita connection, consistently with \cite{Lan86}. It also coincides with the canonical connection of $\bS^3$ as a $SU(2)\times SU(2)$-symmetric space (\cite{KN69}, \S11.3), that is known to solve the Yang-Mills equation \cite{HTS80}. The correspondence between canonical connection and Levi-Civita connection on a Riemannian symmetric space is also a general property \cite[Thm.~3.3]{KN69}.

Such a solution of the Yang-Mills equation on $\bS^3$ can give a non self-dual solution of the Yang-Mills equation on a 4-dimensional manifold. We adopt homogeneous coordinates on $SU(2)$, 
\begin{equation} 
\label{gpu}
SU(2)\,\ni\,g\,=\,\left(\begin{array}{cc} u & -\bar{v} \\ v & \bar{u}\end{array} \right), \qquad\qquad\bar{u}u\,+\,\bar{v}v\,=\,1:
\end{equation}
it is immediate to check that
$$
\alpha\,=\,\lambda\,g^{-1}\dd\,g\,=\,\lambda\,\omega_{a}\,\otimes\,X_a.
$$
If we now consider the projection 
\begin{equation}
\label{proc}
\pi\,:\,\R^4\backslash\{0\}\,\to\,\bS^3\,\qquad\qquad\mathrm{with}\qquad\qquad x_{\mu}\,\mapsto\,\frac{x_{\mu}}{||x||}, 
\end{equation}
(where we  introduced real coordinates $x_{\mu}\,=\,x_{\mu}^*,\,\,\mu\,=\,0,\ldots,3$ and set $u\,=\,x_0\,+\,i\,x_3$ and $v\,=\,x_2\,+\,i\,x_1$, while $||x||^2\,=\,x_{\mu}x_{\mu}$ comes from the euclidean metric on $\R^4$)  we get 
\begin{equation}
\label{pbf}
\pi^{*}(g^{-1})\,=\,\frac{1}{||x||}\,(x_0\,1\,-\,i\,\sigma_a\,x_a),\qquad\quad \pi^*(g)\,=\,\frac{1}{||x||}\,(x_0\,1\,+\,i\,\sigma_a\,x_a)
\end{equation}
where the Pauli matrices are given by $\sigma_a\,=\,-2\,i\,T_a$ and
\begin{equation}
\label{cu4}
\pi^*(F)\,=\,(1\,-\,\lambda)\dd\,A\,=\,\lambda(\lambda\,-\,1)\pi^*(\dd\,g^{-1}\,\wedge\,\dd\,g).
\end{equation} 
An explicit computation shows that, if $\lambda\,=\,1/2$, such a curvature solves the Yang-Mills equation on $\R^4\backslash\{0\}$ with the euclidean metric. This solution is the one presented in \cite{DFF76} as meron.
A simple expression for the meron curvature can be derived as follows.
Let 
$\pi^*(g^{-1}\dd g)\,=\,A_\mu \dd x_\mu$ and $\pi^*(F)\,=\,F_{\mu\nu}\dd x_{\mu}\wedge x_{\nu}$ with $F_{\mu\nu}=\partial_\mu A_\nu-\partial_\nu A_\mu+[A_\mu,A_\nu]$.
Since $F_{\mu\nu}=0$
for $\lambda=1$, $\partial_\mu A_\nu-\partial_\nu A_\mu=-[A_\mu,A_\nu]$
and $F_{\mu\nu}\,=\,\lambda(\lambda-1)[A_\mu,A_\nu]$, that is
\begin{equation*}
F_{\mu\nu}=
\frac{\lambda(\lambda-1)}{||x||^4}\big[\Sigma_{\mu\alpha}x_\alpha,\Sigma_{\nu\beta}x_\beta\big]
\end{equation*}
with $\Sigma_{\mu\nu}=\tau^*_\mu\tau_\nu-\delta_{\mu\nu}$
where 
$\tau_0\,=\,\openone$ and $-i\tau_j$ are the three Pauli matrices.

As one can see from the explicit expression of $F_{\mu\nu}$,
that there is a singularity in zero. One can produce other solutions, with two singularities, using a conformal transformation (the two singularities comes from
$0$ and the point at infinity). One can also check that the -- unlike the instanton solution, defined on the whole $\R^4$ -- action functional
is logarithmically divergent and can be extended to the compactification $\bS^4$, where it has a finite action.

\subsection{A digression: where does the $1/2$ factor come from?}
To clarify what we have only quoted some lines above, consider 
the  principal bundle with total space $P\,=\,G\,\times\,G^{\prime}$, gauge group $H$ with $G\,\sim\,G^{\prime}\,\sim\,H\,=\,{\rm SU(2)}$  and  right principal action $r\,:\,(G\,\times\,G^{\prime},\,H)\,\to\,(G\,\times\,G^{\prime})$ given by:
\beq
r\,:\,(g\,\times\,g^{\prime},\,h)\,\mapsto\,h^{-1}g\,\times\,g^{\prime}\,h.
\label{rpa}
\eeq
For the group $G$ (resp. $G^{\prime}$)  we consider both the set of left invariant vector fields and 1-forms $(L_a,\omega_a)$ (resp. $L_a^{\prime}, \omega_{a}^{\prime})$) and the set of right invariant vector fields and 1-forms $(R_a, \eta_a)$ (resp. $(R^{\prime}_a, \eta_{a}^{\prime})$).   The infinitesimal action of the gauge group is generated by the vertical fields of the fibration, which turn out to be
\beq
V_{a}\,=\,-R_a\,+\,L^{\prime}_a
\label{vfi}
\eeq for $a\,=\,1,\ldots,3$.  On $P$ we consider the gauge invariant Riemannian metrics 
\beq
\label{metrics}
\mathrm{g}\,=\,\sum_{a\,=\,1}^{N}\,(\eta_a\otimes\eta_a\,+\,\omega^{\prime}_a\otimes\omega^{\prime}_a)
\eeq
In terms of such a metrics, a natural choice for a connection is given by the orthogonal splitting of $TP$:
the vertical part, span by $V_a$ in \eqref{vfi},  is fixed by the gauge action,  is orthogonal to the  horizontal part, given by the span of  $H_a\,=\,L_a^{\prime}+R_a$. If we denote by 
 $\{X_a\} $ a basis of the Lie algebra of $H\,\sim\,G$, such a connection can be written as a ${\rm Lie}(G)$-valued 1-form on $P$, namely 
\beq
\label{c1fo}
\Omega\,=\,\frac{1}{2}\,X_a\,\otimes\,(\omega^{\prime}_a\,-\,\eta_a).
\eeq 
In order to obtain, from a connection on the principal bundle $P$, a vector potential on the basis, we need a section of $P$. The principal bundle we are considering is globally trivial, one can write a projection $\pi\,:\,P\,\to\,P/H$
\beq
\pi\,:\,(g\,\times\,g^{\prime})\,\mapsto\,g^{\prime}g
\label{proP}
\eeq
where trivially $P/H\,\sim\,G$. We denote by $M$ the basis of the bundle, and introduce the global section $\sigma\,:\,M\,\to\,P$ via
\beq
\sigma\,:\,m\,\to\,1\,\times\,m
\label{sig}
\eeq 
so that the global trivialization $P\,\sim\,H\,\times\,M$ of the bundle is given by $(h^{-1}, m\,h)$. Via such a global section we define the vector potential  $A\,=\,\sigma^*(\Omega)$ as a ${\rm Lie}(H)$-valued 1-form on the basis $M$. We get, from \eqref{c1fo}:
\beq
\label{pove}
A\,=\,\frac{1}{2}\,X_a\otimes\omega^{\prime}_a
\eeq
which, in terms of the given trivialization, coincides with the meron connection.

\subsection{Hermitian structure on ${\rm T}^*\bS^n$}\label{sec:3.2}
Although $\R^{n+1}\smallsetminus\{0\}$ is
diffeomorphic to $\R^+\times\bS^n$, the metrics is not exactly the product metrics.
For $i=0,\ldots,n$, let $x_i$ be Cartesian coordinates on \mbox{$\R^n\smallsetminus\{0\}$},
\mbox{$r=\left(\sum_ix_i^2\right)\!{}^{1/2}$} the radius, and $\xi_i=x_i/r$ coordinates
on the unit sphere $\bS^n$.
Using the identity $\dd x_i=r\dd \xi_i+\xi_i\dd r$, we write the metric tensor of $\R^{n+1}$
in spherical coordinates:
\begin{equation}\label{eq:warpedsn}
g_{\R^{n+1}}:=\sum \dd x_i\otimes \dd x_i=\dd r\otimes\dd r+r^2\sum\dd \xi_i\otimes\dd \xi_i \;.
\end{equation}
where we used $\sum \xi_i\dd \xi_i=0$ (coming from $\sum \xi_i^2=1$ and the Leibniz rule).
So
\begin{equation*}
g_{\R^{n+1}}=g_{\R^+}+r^2g_{\bS^n}
\end{equation*}
is \emph{almost} a product metric, except for the $r^2$ factor. This affects
the Hermitian structure on differential forms as follows. If we denote
by $(~,~)_m$ the Hermitian structure on $m$-forms on $\R^{n+1}$,
by $(~,~)^{(1)}_m$ the one on $\R^+$ and by $(~,~)^{(2)}_m$ the one on $\bS^n$,
we have
\begin{equation}\label{eq:herRS}
(\xi_1\otimes\xi_2,\eta_1\otimes\eta_2)_m :=
r^{-2k}(\xi_1,\eta_1)_{m-k}^{(1)}\otimes (\xi_2,\eta_2)_k^{(2)}
\end{equation}
for all $\xi_1,\eta_1\in\Lambda^{m-k}(\R^+)$ and $\xi_2,\eta_2\in\Lambda^k(\bS^n)$,
where now $r$ is the function that associates to a point in $\R^{n+1}$ its radius.
Exactly like for a Cartesian product, $\Lambda^m(\R^{n+1}\smallsetminus\{0\})=\bigoplus_{k=0}^m\Lambda^{m-k}(\R^+)\otimes\Lambda^k(\bS^n)$ and with the notation
\begin{equation*}
\dd=\sum \dd x_i\frac{\partial}{\partial x_i} \;,\qquad
\dd_1=\dd r\frac{\partial}{\partial r} \;,\qquad
\dd_2=\sum \dd \xi_i\frac{\partial}{\partial \xi_i} \;,
\end{equation*}
one has $\dd=\dd_1\otimes 1+1\otimes\dd_2$ on functions (extended to forms with the graded
Leibniz rule). The volume forms are related by
$\dd^{n+1}x=r^n\dd r\,\dd^n\theta$, where $\dd^n\theta$ is the volume form on $\bS^n$.
For a product $M_1\times M_2$ of noncommutative spaces, with metrics of the above type,
we will show that solutions of Yang-Mills on $M_2$ give solutions of Yang-Mills on the product.

\section{Warped product of noncommutative spaces}\label{sec:3}

In the previous section we worked out a particular case of a general construction
called warped product.
The \emph{warped product} \cite{BoN69} $M_1\times_f M_2$ of two Riemannian manifolds $(M_1,g_1)$ and $(M_2,g_2)$, with warping function $f:M_1\to\R^+$, is the product manifold
$M=M_1\times M_2$ equipped with the metric
\begin{equation*}
g=g_1+f^2g_2 \;.
\end{equation*}
Warped product manifolds play an interesting role in both differential geometry and physics (see e.g.~\cite{BEP82,ONei83}). The usual product manifold metric give trivial examples of them, while a first non trivial example is given by the  decomposition in polar coordinates systems of a euclidean $\R^n$ space.
Several Riemannian symmetric spaces are examples of warped products, e.g.~the hyperbolic space $H^{n+1}$ can be decomposed as $\R\times_{f}\R^n$ with $f(t)=e^t$, and every revolution surface which does not intersect the revolution axis is isometric to $M_1\times_f\bS^1$, with
$M_1$ the generating curve and $f(x)$ the distance from the axis.
The warped product makes sense in pseudo-Riemannian geometry as well (in fact, this is the cases of physical interest).  Robertson-Walker spaces (modelling an expanding universe in cosmology) correspond to the case of $M_2$  a Riemannian manifold of constant curvature, and $M_1$  an interval of $\R^+$ with pseudo-Riemannian metric $-\,\dd t\otimes\dd t$ ($t$ is a time coordinate).
The Schwarzschild space-time (modelling the outer space around a massive star or a black hole) is indeed a  warped product manifold.

There are several generalizations (e.g.~local warped products, multi-warped products).
In the present paper we consider a noncommutative geometric version of this construction. We will explain how the Hodge star on a warped product is related to the ones on the two factors, and use this to prove that the pullback of a solution of Yang-Mills on $M_2$ solves Yang-Mills on $M$.

\subsection{Product of calculi}
For $i=1,2$, let $(\Lambda_i^\bullet,\dd_i)$ be a differential $*$-calculus
of dimension $n_i$ (i.e.~$\Lambda^k_i=0\;\forall\;k>n_i$) over $\A_i=\Lambda_i^0$, and denote by $\wedge_i$ the product. Assume, as in \S\ref{sec:1.4.2}, that $\Lambda^{n_i}_i$
is a free right $\A_i$-module of rank $1$, with basis element $\tau_i$, and that each $\Lambda_i^k$ is a self-dual right-module with Hermitian structure $(~,~)_k^{(i)}$.
We also assume that $(a,b)_0^{(i)}=(\tau_ia,\tau_ib)_{n_i}^{(i)}=a^*b$ for all $a,b\in\A_i$.
A Hodge star operator $\ast_H^{(i)}:\Lambda^{n_i-k}_i\to\Lambda^k_i$ is uniquely defined by
\eqref{eq:2.2}.

The analogue of the Cartesian product of smooth manifolds is the 
tensor product of differential complexes (see e.g.~\S1.0.14 of \cite{Lod92}).
The algebra $\A=\A_1\otimes\A_2$ plays the role of the algebra of ``functions'' on the product space.
A differential $*$-calculus $(\Lambda^\bullet,\dd)$ on $\A$, with
dimension $n=n_1+n_2$, is given by
\begin{equation}\label{eq:1}
\Lambda^m=\bigoplus_{k=0}^{m}\Lambda_1^{m-k}\otimes\Lambda^k_2 \;.
\end{equation}
This is a graded associative algebra with product
\begin{equation*}
(\xi_1\otimes\xi_2)\wedge (\eta_1\otimes\eta_2)=(-1)^{\mathrm{dg}(\xi_2)\mathrm{dg}(\eta_1)}
(\xi_1\wedge_1\eta_1)\otimes(\xi_2\wedge_2\eta_2)
\end{equation*}
for all $\xi_i,\eta_i\in\Lambda_i^\bullet$, $i=1,2$.
If we set  the differential as
\begin{equation}\label{eq:diff}
\dd(\xi_1\otimes\xi_2)=\dd_1\xi_1\otimes\xi_2+(-1)^{\mathrm{dg}(\xi_1)}\xi_1\otimes\dd_2\xi_2 \;,
\end{equation}
the graded Leibniz rule is automatically satisfied, while defining the involution as
\begin{equation*}
(\xi_1\otimes\xi_2)^*=\xi_1^*\otimes\xi_2^* \;,
\end{equation*}
we get a differential $*$-calculus on $\A$.

\subsection{Warped product of Hermitian structures}
\label{ss:wp}

A right-module Hermitian structure on $\Lambda^m$ is defined as follows. As usual, we declare  $\Lambda^m$ and $\Lambda^{m'}$
to be are orthogonal if $m\neq m'$ and declare also  $\Lambda_1^{m-k}\otimes\Lambda^k_2\subset\Lambda^m$ and $\Lambda_1^{m-h}\otimes\Lambda^h_2\subset\Lambda^m$
to be orthogonal too if $k\neq h$, while for $h=k$, inspired by \eqref{eq:herRS}, we set:
\begin{equation*}
(\xi_1\otimes\xi_2,\eta_1\otimes\eta_2)_m =
\rho_k^{-2}(\xi_1,\eta_1)_{m-k}^{(1)}\otimes (\xi_2,\eta_2)_{k}^{(2)}
\end{equation*}
for all $\xi_1,\eta_1\in\Lambda_1^{m-k}$ and $\xi_2,\eta_2\in\Lambda_2^k$,
where $\rho_0,\ldots,\rho_{n_2}\in\A_1$ are some fixed positive invertible elements.
Notice that $\Lambda^n=\Lambda_1^{n_1}\otimes\Lambda_2^{n_2}$ is a free right $\A$-module
of rank $1$. As a basis element for it we choose:
\begin{equation*}
\tau=\tau_1\rho_{n_2}\otimes\tau_2 \;,
\end{equation*}
which  is real whenever $\tau_1\rho_{n_2}$ and $\tau_2$ are. A corresponding Hodge star operator $\ast_H:\Lambda^k\to\Lambda^{n-k}$ is defined by \eqref{eq:2.2}.
In the example of warped product of two Riemannian manifolds, it is $\rho_k=f^k$
for all $k=0,\ldots,n_2$, where $f$ is the warping function.

\begin{prop}
For  $\xi_1\otimes\,\xi_2\,\in\,\Lambda^{m-k}_1\otimes\Lambda^k_2$, it is
\begin{equation}\label{eq:prodHodge}
\ast_H(\xi_1\otimes\,\xi_2)\,=\,(-1)^{k(n_1-m+k)}\ast_H^{(1)}(\xi_1)\rho^{-1}_{n_2}\rho_{n_2-k}^2\otimes\ast_H^{(2)}(\xi_2),
\end{equation}
where $\ast_{H}^{(i)}\,:\,\Lambda_i^{k}\,\to\,\Lambda_i^{n_i-k}$ are the Hodge duality previously defined on $\Lambda^{\bullet}_i$ for $i\,=\,1,2$.
\end{prop}

\begin{proof}
By construction, for all $\xi_1\in\Lambda_1^{m-k},\eta_1\in\Lambda_1^{n_1-m+k},
\xi_2\in\Lambda_2^k,\eta_2\in\Lambda_2^{n_2-k}$, we have:
\begin{align*}
&\tau\,(\,\ast_H\,(\xi_1\,\otimes\,\xi_2)\,,\,\eta_1\,\otimes\,\eta_2\,)^{n-m}\, 
\\[8pt] &\qquad
=\,
(\xi_1\,\otimes\,\xi_2)^*\,\wedge\, (\eta_1\,\otimes\,\eta_2)
\\[8pt] 
&\qquad
=\,(\xi_1^*\otimes\xi_2^*)\,\wedge \,(\eta_1\otimes\eta_2)
\\[8pt] & \qquad
=\,(-1)^{k(n_1-m+k)}\,(\xi_1^*\wedge_1\eta_1)\otimes(\xi_2^*\wedge_2\eta_2)
\\[8pt] & \qquad
=(-1)^{k(n_1-m+k)}\,\tau_1\,(\,\ast_H^{(1)}\xi_1\,,\,\eta_1\,)_{n_1-m+k}^{(1)}\,\otimes\, \tau_2\,(\,\ast_H^{(2)}\xi_2\,,\,\eta_2\,)_{n_2-k}^{(2)}
\\[8pt] &\qquad
=(-1)^{k(n_1-m+k)}\,\big\{\,\tau_1\,\rho_{n_2}\,\otimes\,\tau_2\,\big\}\,\wedge\,
\big\{\,\rho^{-1}_{n_2}\,(\,\ast_H^{(1)}\xi_1\,,\,\eta_1\,)_{n_1-m+k}^{(1)}\,\otimes\, (\,\ast_H^{(2)}\xi_2\,,\,\eta_2\,)_{n_2-k}^{(2)}\,\big\}
\\[8pt] & \qquad
=(-1)^{k(n_1-m+k)}\,\tau\,
\big\{\,\rho_{n_2-k}^{-2}\,(\,\ast_H^{(1)}\,(\,\xi_1\,)\,\rho^{-1}_{n_2}\,\rho_{n_2-k}^2\,
,\,\eta_1\,)_{n_1-m+k}^{(1)}\,\otimes\, (\,\ast_H^{(2)}\,\xi_2\,,\,\eta_2\,)_{n_2-k}^{(2)}\big\}
\\[8pt] &\qquad
=(-1)^{k(n_1-m+k)}\,\tau\,(\,\ast_H^{(1)}\,(\xi_1)\,\rho^{-1}_{n_2}\,\rho_{n_2-k}^2\,\otimes\,
\ast_H^{(2)}\,\xi_2\,,\,\eta_1\,\otimes\,\eta_2\,)_{n-m}.
\end{align*}
From the non-degeneracy of the Hermitian structure we then derive \eqref{eq:prodHodge}.
\end{proof}

\subsection{Yang-Mills and warped products}
\label{ss:ym}

We now prove that, given a solution of the Yang-Mills equations over a suitable class of manifolds -- that we denote by $M_2$ -- its pullback via 
the projection
$M=M_1\times_fM_2\to M_2$ solves the Yang-Mills equation on the warped product $M$. We give the proof in the general
framework of previous section, considering  connections on free modules (as
is the case of merons), i.e. operators  of the form (see \eqref{ccu}) $\nabla=\dd+\alpha$ where $\alpha$ is a suitable matrix whose entries are $1$-forms.

We start by noticing that the pullback of $\alpha_2\,\in\, \mathbb{M}_N(\Lambda^1_2)$ is   the element $1\,\otimes\, \omega_2\,\in\, \mathbb{M}_N(\Lambda^1)$.

\begin{prop}
The connection $\nabla_2=\dd_2+\omega_2$ solves the Yang-Mills equation on the free module $\A_2^N$
if and only if $\nabla=\dd+1\otimes \omega_2$ solves the Yang-Mills equation on $\A^N$.
\end{prop}

\begin{proof}
The curvature of $\nabla$ is $F\,=\,1\,\otimes\, F_2$, with $F_2$
the curvature of $\nabla_2$. From \eqref{eq:prodHodge} and Prop.~\ref{prop:list}(ii):
\begin{equation*}
\ast_H\,F \,=\,\tau_1\,\rho^{-1}_{n_2}\,\rho_{n_2-2}^2\,\otimes\,\ast_H^{(2)}\,F_2
\end{equation*}
Since $\dd_1$ is zero on $\Lambda_1^{n_1}$, from \eqref{eq:diff}:
\begin{equation*}
\dd(\ast_H\,F)\,=\,(-1)^{n_1}\,\tau_1\,\rho^{-1}_{n_2}\,\rho_{n_2-2}^2\,\otimes\,\dd_2(\ast_H^{(2)}\,F)
\end{equation*}
and
\begin{equation*}
[\nabla,\star_H\,F]\,=\,(-1)^{n_1}\,\tau_1\,\rho^{-1}_{n_2}\,\rho_{n_2-2}^2\,\otimes\, [\nabla_2,\ast_H^{(2)}\,F_2] \;.
\end{equation*}
Thus $[\nabla,\star_H\,F]\,=\,0$ if and only if $[\nabla_2,\ast_H^{(2)}\,F_2]\,=\,0$.
\end{proof}

We stress that for $n=4$, since the element $F\,\in\,\Lambda_1^0\,\otimes\,\Lambda_2^2$
and the element  $\ast_H\,F\,\in\,\Lambda_1^{n_1}\,\otimes\,\Lambda_2^{n_2-2}$ are in two  orthogonal
subspaces, $F$ can not be an eigenvector of the Hodge duality operator (unless
$n_1=0$).
In the example of $\R^4\smallsetminus\{0\}=\R^+\,\times_{r}\,\bS^3$, any solution of
Yang-Mills on $\bS^3$ gives a solution on $\R^4\smallsetminus\{0\}$ that is not
(anti) self-dual.

\section{Solutions of the Yang-Mills equations over $\SU$}
\label{s:YMq}

As quantum group $\SU$ (we adopt the notations and the conventions in \cite{KS97}) we consider the  polynomial  unital  $*$-algebra $\ASU=(\SU,\Delta,S,\eps)$ generated by elements $a$ and $c$ such that, using the matrix notation
\beq
\label{Us}
 u = 
\left(
\begin{array}{cc} a & -qc^* \\ c & a^*
\end{array}\right) , 
\eeq
the Hopf algebra structure can be expressed as
$$
uu^*=u^*u=1,\qquad\qquad \,\Delta\, u = u \otimes u ,  \qquad\qquad S(u) = u^* , \qquad\qquad \eps(u) = 1.
$$
The deformation parameter
$q\in\IR$ is restricted  without loss of generality  to the interval $0<q<1$. 
The quantum universal envelopping algebra $\su$ is the unital Hopf $*$-algebra
generated the four elements $K^{\pm 1},E,F$, with $K K^{-1}=1$ and the
relations: 
\beq 
K^{\pm}E=q^{\pm}EK^{\pm}, \qquad\qquad 
K^{\pm}F=q^{\mp}FK^{\pm}, \qquad\qquad  
[E,F] =\frac{K^{2}-K^{-2}}{q-q^{-1}} . 
\label{relsu}
\eeq 
The $*$-structure is
$K^*=K, \,  E^*=F $,
while the Hopf algebra structures are 
\begin{align*}
&\Delta(K^{\pm}) =K^{\pm}\otimes K^{\pm}, \quad
\Delta(E) =E\otimes K+K^{-1}\otimes E,  \quad 
\Delta(F)
=F\otimes K+K^{-1}\otimes F,
\\ &\qquad S(K) =K^{-1}, \quad
S(E) =-qE, \quad 
S(F) =-q^{-1}F \\ 
&\qquad\qquad\varepsilon(K)=1, \quad \varepsilon(E)=\varepsilon(F)=0.
\end{align*}
The non degenerate Hopf algebra pairing between the two algebras above is  the $*$-compatible bilinear mapping $\hs{~}{~}:\su\times\ASU\to\IC$   given on the generators by
\begin{align}
\langle K^{\pm},a\rangle\,=\,q^{\mp 1/2},& \qquad\qquad  \langle K^{\pm},a^*\rangle\,=\,q^{\pm 1/2} \nn \\   \langle E,c\rangle\,=\,1, &\qquad\qquad \langle F,c^*\rangle\,=\,-q^{-1} \label{ndp}
\end{align}
with all other couples of generators pairing to zero. The algebra $\su$ is recovered as a $*$-Hopf subalgebra in the dual algebra $\ASU^o$, the largest Hopf $*$-subalgebra contained in the dual vector space  $\ASU^{\prime}$.
The $*$-compatible canonical commuting actions of $\su$ on $\ASU$ are
$$
h \lt x = x_{(1)} \,\hs{h}{x_{(2)}}, \qquad\qquad
x \rt  h = \hs{h}{x_{(1)}}\, x_{(2)}. 
$$
The first order differential calculus (FODC)  $(\dd, \Lambda^1)$  we are going to consider on  $\SU$ was introduced by Woronowicz \cite{Wor87}; $\Lambda^1$ is a three dimensional left covariant free $\ASU$-bimodule.
A basis of left-invariant 1-forms ($\Lambda^1_{{\rm inv}}\,\subset\,\Lambda^1$) is 
\begin{align}
&\qquad\qquad\omega_{z} =a^*\dd a+c^*\dd c ,\nn \\
&\qquad\qquad \omega_{-} =c^*\dd
a^*-qa^*\dd c^*, \nn \\ 
&\qquad\qquad\omega_{+} =a\dd c-qc \dd a.
\label{q3dom}
\end{align}
It is  a $*$-calculus, with  $\omega_{-}^*=-\omega_{+}$ and
$\omega_{z}^*=-\omega_{z}$.  
The  basis of the quantum tangent space  $\mathcal{X}\,\subset\su$ dual to $\Lambda_{{\rm inv}}$ \eqref{q3dom}  
 is given by
\begin{align} 
&\qquad\qquad\qquad\qquad X_{-}=q^{-1/2}FK, \nn \\
&\qquad\qquad\qquad \qquad X_{+}=q^{1/2} EK, \nn \\ 
&\qquad\qquad\qquad \qquad X_{z}= (1-q^{-2})^{-1}\left(1-K^4\right). \label{Xv3}
\end{align}
so that one can write exact 1-forms as $\dd x=\sum_{a=\pm,z}(X_{a}\lt x)\omega_{a}$
with $x\,\in\,\ASU$. 
This FODC has a non canonical braiding $\sigma$ on $\Gamma^{\otimes 2}$ \cite{Hec01}, which is as follows 
\begin{align}
\sigma(\omega_{a}\otimes\omega_{a})=\omega_{a}\otimes\omega_{a},&\qquad\qquad a=\pm,z
\nn \\
\sigma(\omega_{-}\otimes\omega_{+})=(1-q^{2})\omega_{-}\otimes\omega_{+}+q^{-2}\omega_{+}\otimes\omega_{-}, & \qquad\qquad \sigma(\omega_{+}\otimes\omega_{-})=q^4\omega_{-}\otimes\omega_{+}, \nn \\
\sigma(\omega_{-}\otimes\omega_{z})=(1-q^{2})\omega_{-}\otimes\omega_{z}+q^{-4}\omega_{z}\otimes\omega_{-}, &
\qquad\qquad \sigma(\omega_{z}\otimes\omega_{-})=q^6\omega_{-}\otimes\omega_{z}, \nn \\
\sigma(\omega_{z}\otimes\omega_{+})=(1-q^{2})\omega_{z}\otimes\omega_{+}+q^{-4}\omega_{+}\otimes\omega_{z},& \qquad\qquad
\sigma(\omega_{+}\otimes\omega_{z})=q^6\omega_{z}\otimes\omega_{+}.
\label{brai}
\end{align}
This braiding allows to introduce both a wedge product among 1-forms (and then to build the complete exterior algebra $\Lambda^{\bullet}\,=\,\oplus_{k=0}^3\Lambda^{ k}$) and a quantum commutator $[~,~]_{\sigma}$ in $\mathcal{X}$ satisfying a peculiar Jacobi identity. The 
exterior wedge product satisfies
$\omega_{a}\wedge\omega_{a}=0, \quad (a=\pm,z)$ and  
\begin{align}
&\qquad\qquad \omega_{-}\wedge\omega_{+}+q^{-2}\omega_{+}\wedge\omega_{-}=0, \nn \\
&\qquad\qquad
\omega_{z}\wedge\omega_{\mp}+q^{\pm4}\omega_{-}\wedge\omega_{z}=0,
\label{commc3} 
\end{align}
while for the quantum commutator one has
\begin{align}
&\qquad\qquad[X_a, X_a]_{\sigma}\,=\,0 \qquad\qquad a\,=\,(\pm,z)\nn \\
&\qquad\qquad[X_-,X_+]_{\sigma}\,=\,-q^4[X_+,X_-]_{\sigma}\,=\,q^2X_z \nn \\ 
&\qquad\qquad [X_-,X_z]_{\sigma}\,=\,-q^6[X_z,X_-]_{\sigma}\,=\,-q^4(1+q^2)X_- \nn \\
&\qquad\qquad [X_+,X_z]_{\sigma}\,=\,-q^{-6}[X_z,X_+]_{\sigma}\,=\,(1+q^{-2})X_+.
\label{qcv}
\end{align}
If we write the above formulas as $[X_{a}, X_b]_{\sigma}\,=\,f_{ab}^cX_c$ then a Maurer Cartan equation is valid, namely
\beq
\label{mcq}
\dd\,\omega_a\,=\,-\frac{1}{1+q^2}\,f_{bc}^a\,\omega_{b}\wedge\,\omega_c\,=\,-\frac{1}{\lambda}\,f_{bc}^a\,\omega_{b}\wedge\,\omega_c.
\eeq 
where the last equality defines the coefficient $\lambda$: this relation is the counterpart of the classical \eqref{mccl} within such a  quantum formalism.

\subsection{A Hodge duality operator over $\SU$}
\label{ss:hsq}

In order to introduce a meaningful Hodge duality operator over the exterior algebra $\Lambda^{\bullet}$ described above we consider 
the  free (right) self-dual modules $\Lambda^k,\,(k\,=\,0,\ldots, 3)$ and set  the hermitian structure on a basis in $\Lambda^1$ as
\beq
(\omega_+,\omega_+)_1\,=\,\alpha,\qquad\qquad(\omega_-,\omega_-)_1\,=\,\beta,\qquad\qquad(\omega_z,\omega_z)_1\,=\,\gamma
\label{sth1}
\eeq
with $\alpha, \beta, \gamma$ complex numbers such that $\alpha\,\beta\,\gamma\,\neq\,0$. Following \cite{ale11, ale12}, the antisymmetriser operators of the differential calculus (those coming from the braiding in \eqref{brai}) allow to extend it to a hermitian structure on $\Lambda^k$ for $k\,>\,1$. Given the hermitian volume form $\tau\,=\,\omega_-\wedge\omega_+\wedge\omega_z$ we introduce a Hodge duality $\ast_H\,:\,\Lambda^{k}\,\to\,\Lambda^{3-k}$ via the equation \eqref{eq:2.2}
\beq 
\tau(\ast_H\,\xi, \xi^{\prime})_{3-k}\,=\, \xi^*\wedge\,\xi^{\prime},
\label{defh}
\eeq
which trivially gives $\ast_H(1)\,=\,\tau$ and $\ast_H(\tau)\,=\,1$. The non trivial terms are given by:
\begin{align}
&\ast_H(\omega_-)\,=\,(\lambda\,q^4/2\,\beta\,\gamma)\,\omega_-\wedge\omega_z, \nn \\
&\ast_H(\omega_+)\,=\,-(\lambda\,q^{-6}/2\,\alpha\,\gamma)\,\omega_+\wedge\omega_z, \nn \\
&\ast_H(\omega_-)\,=\,-(\lambda/2\,\alpha\,\beta)\,\omega_-\wedge\omega_+ 
\label{hsf2}
\end{align}
and 
\begin{align}
&\ast_H(\omega_-\wedge\omega_z)\,=\,(q^2/\beta)\,\omega_-, \nn \\
&\ast_H(\omega_+\wedge\omega_z)\,=\,-(1/\alpha)\,\omega_+, \nn \\ 
&\ast_H(\omega_-\wedge\omega_+)\,=\,-(1/\gamma)\,\omega_z, 
\label{hsf1}
\end{align}  
where the parameter  $\lambda$ was defined in \eqref{mcq}. The two equirements that  (i) the  degeneracy of the operator $(\ast_{H})^2$ on each $\Lambda^k$  is  suitably compatible with the braiding of the differential calculus and (ii) the action of $\ast_H$ commutes with that of the complex conjugation, are taken as the  
 definition of  reality and symmetry for the scalar product given in \eqref{sth1}.  We have 
 \beq
 \label{coH}
 \R\,\ni\,\alpha\,=\,q^{-6}\beta\,\neq\,0,\qquad\qquad\R\,\ni\,\gamma\,\neq\,0
 \eeq
 and we see that   such a Hodge duality is not one of those presented in \cite{ale12}, but gives the same notions of symmetry and reality for the corresponding scalar product \eqref{sth1}.

The relations \eqref{hsf2} and \eqref{hsf1}, together with the conditions \eqref{coH}, give us a meaningful Hodge duality operator on $\Lambda^{\bullet}$ over $\SU$. In order to get a solution of the Yang-Mills equations over $\SU$ which mimics the classical one described in section \ref{sec:YMS3} we need to select the Hodge duality that would extend, within the quantum setting,  the one -- given in \eqref{Hcl} -- associated to the Cartan-Killing metrics on the classical $\bS^3$. We notice from \eqref{mccl}  that on $\bS^3$ the relation $\ast_H(\omega)\,=\,-\,\dd\omega$ holds for any $\omega\,\in\,\Lambda^1_{{\rm inv}}$, and also that such a relation completely characterize a conformal class of metric tensors on $\bS^3$, namely those to which the Cartan-Killing metrics belongs. It seems then natural to select a class of Hodge duality operators on $\SU$ by setting:
\beq
\label{ckq}
\ast_{H}(\omega_{a})\,=\,\mu\,\dd\omega_a
\eeq
on a basis of $\Lambda^1_{{\rm inv}}$ for $\SU$ with a constant $0\,\neq\,\mu\,\in\,\R$. This relation turns out to be compatible with the constraints \eqref{coH}, giving 
\beq
\label{exc}
(1\,+\,q^{-2})\gamma\,=\,\alpha\,=\,q^{-6}\beta\,\neq\,0,
\eeq
 providing  us what we assume to be a Hodge duality operator corresponding to a Cartan-Killing structure on $\SU$.
 
 \subsection{Yang-Mills equations over $\SU$}
 \label{ss:eqq}
 
Together with the commutator structure \eqref{qcv}  given by the braiding,  we consider $\mathcal{X}$  as the quantum Lie algebra associated to the Woronowicz's calculus. Its spin $j\,=\,0, \,1/2,\,1,\ldots$ representations on $\C^n$, with $n\,=\,2j+1$, are well known \cite{KS97}. The spin $j\,=\,1/2$ representation is given for instance by: 
\beq
X_{z}\,=\,\left(\begin{matrix} 1 & 0 \\ 0 & -q^2\end{matrix}\right), \qquad
X_{-}\,=\,\left(\begin{matrix} 0 & 1 \\ 0 & 0\end{matrix}\right), \qquad
X_{+}\,=\,\left(\begin{matrix}  0 & 0 \\ 1 & 0 \end{matrix}\right).
\label{repr}
\eeq 
We define the free $\ASU$-module $\mathcal{E}\,=\,\ASU^{\otimes2}\,=\,\C^2\,\otimes\,\ASU$ and  the vector potential  (the connection 1-form, see \eqref{ccu}) in analogy to \eqref{calpha}:
 \beq 
\mathbb{M}_2(\ASU)\,\otimes_{\ASU}\Lambda^1\,\ni\,\mathfrak{A}\,=\,\epsilon\,\sum_j\,X_j\otimes\,\omega_j\,=\,\epsilon\left(\begin{matrix} \omega_z & \omega_- \\ \omega_+  & -q^2\,\omega_z \end{matrix}\right).
\label{poym}
\eeq 
For $\epsilon\,=\,1$ such a vector potential gives the Maurer-Cartan form on $\SU$, with  $\mathfrak{A}\,=\,u^{-1}\dd u$ where $u$ is given in \eqref{Us}. This means that we have:  
\beq
\epsilon\,\dd\,\mathfrak{A}\,=\,-\,\mathfrak{A}\wedge\mathfrak{A}
\label{qua}
\eeq
while the curvature (see \eqref{dcm}) is 
\beq 
F\,=\,\dd\,\mathfrak{A}\,+\,\mathfrak{A}\wedge\mathfrak{A}\,=\,(1-\epsilon)\dd\,\mathfrak{A}.
\label{cuym}
\eeq 
From \eqref{ckq} it is now immediate to obtain:
\beq 
\ast_H\,F\,=\,(1-\epsilon)\mu\,\mathfrak{A}
\label{defi}
\eeq
and 
\beq
D(\ast_H \,F)\,=\,\mu\,(1-\epsilon)\,(1-2\epsilon)\,\dd\,\mathfrak{A}.
\label{meq}
\eeq
If we choose $\epsilon\,=\,1/2$ it is clear that we have a solution of the Yang-Mills equations on $\SU$, which is the natural analogue of the classical one on $\bS^3$.

\subsection{The extension to $\R_{q}^{4}\backslash\{0\}$}
In the classical setting we extended the meron solution to $\R^4\backslash\{0\}$, obtaining the solution to Yang-Mills equations given in \cite{DFF76}. We generalise this procedure to the quantum setting. 
We consider $\A(\R^4_q)$ to be the unital involutive $*$-algebra generated by (again, as in section \ref{s:YMq}, we consider $0<q<1$)
$x_1,x_2,x_3,x_4$ with commutation relations:
\begin{align}
&x_ix_j    =qx_jx_i \qquad\quad\forall\; i<j,\;i+j\neq 5, \nn \\
&x_2x_3=x_3x_2,  \qquad\quad
x_1x_4-x_4x_1 =(q^{-1}-q)x_2x_3.
\label{cr4}
\end{align}
The $*$-structure is given by $x_1^*=qx_4$, $x_2^*=x_3$. This algebra has been considered in \cite{Maj94}. 
It has a positive central group-like element, namely $\mathcal{D}\,=\,qx_1x_4+q^2x_2x_3$. 
The localization $\R^4_q\backslash\{0\}$, space underlying the quantum group $GL_q(1,\Q)$ of invertible quaternions (\cite{BL12, Fio07, Fior7}) is obtained
by adding the square root of such an element as an invertible  central generator $r^{-1}=(r^{-1})^*$ to the algebra $\A(\R^4_q)$ and a relation:
\begin{equation}
\label{r2}
(qx_1x_4+q^2x_2x_3)r^{-2}=1 \;.
\end{equation}
The Hopf algebra structure of $GL_q(1,\Q)$ is obtained by imposing that the matrix
$$
Q\,=\,\bigg(\!\begin{array}{rr}
x_1 & -qx_3 \\[1pt]
x_2 &  qx_4
\end{array}\!\bigg)
$$
is a corepresentation, namely $\Delta\,Q\,=\,Q\,\otimes \,Q,\, S(Q)\,=\,Q^*,\,\varepsilon(Q)\,=\,1$.
A surjective Hopf $*$-algebra morphism $\pi\,:\,\A(GL_q(1, \Q))\,\to\, \A(\rm{SU_q(2)})$ is given by  
$$
\pi(Q)\,=\,u,\qquad\qquad
\pi( r )\,=\,1,
$$ with the matrix $u$ defined in \eqref{Us}.  This means that ${\rm SU_q(2)}$ is a quantum subgroup of $GL_q(1, \Q)$. An injective Hopf $*$-algebra morphism is $i\,:\,\A({\rm SU_q(2)})\,\hookrightarrow\,\A(GL_q(1, \Q))$, given by 
$$
i(u)\,=\,r^{-1}Q.
$$
To pull the solution of the Yang-Mills  equations on ${\rm SU_q(2)}$ given in \eqref{poym} with $\epsilon\,=\,1/2$ back to a solution of the Yang-Mills equations  on $\R^4_q\backslash\{0\}\,\sim\,GL_q(1,\Q)$ we use the formalism  outlined in section \ref{sec:3} on warped products of non commutative spaces.

Following the notations of section \ref{sec:3},  $\mathcal{A}_1$ is the radial subalgebra in $\A(GL_q(1,\Q))$ generated by $(r, r^{-1})$, with $\Lambda_1$  the free one dimensional bimodule over $\A_1$ generated by $\dd r$ with 
$\dd r^{-1}\,=\,-r^{-2}\dd r$. We also have $\mathcal{A}_2\,=\,\ASU$ and then
$\Lambda_2$ is the exterior algebra $\Lambda^{\bullet}$ built over the Woronowicz' calculus.  
The bimodule $\Lambda^1$ of 1-forms on $GL_q(1,\Q)$ is generated by $\{\dd x_a, \dd r\}$ with $a\,=\,1,\ldots,4$. The exact 1-forms (that we collectively denote by $\dd Q$)  are defined in terms of the map $\pi$, following \eqref{eq:diff}, as
\beq
\label{egq}
\dd Q\,=\,(\dd r)\pi(Q)\,+\,r\,\dd(\pi(Q)):
\eeq 
in order to simplify the notations, we omitted in this expression the tensor product and denoted  by $\dd(\pi(Q))$ the exact forms coming from the Woronowicz calculus on $\SU$.
The first order differential calculus  $(\Lambda^1, \dd)$ turns out to be left-covariant with respect to the coproduct of $GL_q(1,\Q)$. A basis of left-invariant forms is given by the four elements 
\beq
\label{foi}
\phi\,=\,i(u^*)\,r\dd Q\,=\,Q^*\dd Q.
\eeq
Comparing this case to the theory outlined in sections \ref{ss:wp} and \ref{ss:ym}  we have $n_1\,=\,1, n_2\,=\,3$ and $\mathcal{A}_1\,\ni\,\rho_k\,=\,r^k$  with $\tau_1\,=\,\dd r$ and $\tau_2\,=\,\omega_-\wedge\omega_+\wedge\omega_z$ so that the volume form on $\R^4_q\backslash\{0\}$ is $\tau\,=\,r^3\,\dd r\wedge\,\omega_-\wedge\omega_+\wedge\omega_z$. 
Given \eqref{poym}, we have that
\beq
i^*(\mathfrak{A})\,=\,\frac{1}{2}\,r^{-1}\left(
\begin{array}{cc} qx_4\dd(r^{-1}x_1)\,+\,x_3\dd(r^{-1}x_2) &   x_3\dd(r^{-1}qx_4)\,-\,q^2x_4\dd(r^{-1}x_3) \\ x_1\,\dd(r^{-1}x_2)\,-\,q\,x_2\,\dd(r^{-1}x_1)
 & -q^2\{qx_4\dd(r^{-1}x_1)\,+\,x_3\dd(r^{-1}x_2)\}
\end{array}
\right)
\eeq
will be the quantum meron solution on $\R^4_q\backslash\{0\}$.

\end{document}